\documentclass[11pt]{article}
\usepackage{amsmath,amssymb,amsthm,xspace}
\usepackage{tikz}
\usepackage{expdlist}
\usepackage[utf8]{inputenc}
\usepackage{authblk}
\usepackage{fullpage}
\usepackage{comment}
\usepackage{enumerate,amssymb}
\usepackage[hidelinks]{hyperref}
\usetikzlibrary{decorations.pathreplacing}

\newcommand{\Oh}{\mathcal{O}}

\newcommand{\eps}{\varepsilon}
\newcommand{\sub}{\subseteq}

\newcommand{\pint}{\mathbb{Z}_{\ge 0}}
\newcommand{\sm}{\setminus}

\newcommand{\F}{\mathcal{F}}

\DeclareMathOperator*{\ptw}{pw}
\newtheorem{theorem}{Theorem}
\newtheorem{corollary}[theorem]{Corollary}
\newtheorem{lemma}[theorem]{Lemma}

\newtheorem{definition}[theorem]{Definition}
\newtheorem{observation}[theorem]{Observation}
\newtheorem{claim}[theorem]{Claim}

\newcommand{\defproblem}[4]{
  \vspace{1mm}
\noindent\fbox{
  \begin{minipage}{0.96\textwidth}
  #1 \\
  {\bf{Input:}} #2  \\
  {\bf{Find:}} #3 \\
  {\bf{Maximize:}} #4
  \end{minipage}
  }
  \vspace{1mm}
}

\begin{document}
\title{Approximating Upper Degree-Constrained Partial Orientations\footnote{This work is partially supported by Foundation for Polish Science grant HOMING PLUS/2012-6/2.}}

\author{Marek Cygan}
\author{Tomasz Kociumaka}

\affil{Institute of Informatics, University of Warsaw\\ \texttt{[cygan,kociumaka]@mimuw.edu.pl}}

\date{\empty}
\maketitle

\begin{abstract}
In the {\sc Upper Degree-Constrained Partial Orientation} problem
we are given an undirected graph $G=(V,E)$,
together with two degree constraint functions $d^-,d^+ : V \to \mathbb{N}$.
The goal is to orient as many edges as possible, in such a way
that for each vertex $v \in V$ the number of arcs entering $v$ is at most $d^-(v)$,
whereas the number of arcs leaving $v$ is at most~$d^+(v)$.
This problem was introduced by Gabow~[SODA'06], who proved  it to be MAXSNP-hard (and thus APX-hard).
In the same paper Gabow presented an LP-based iterative rounding $4/3$-approximation algorithm.

Since the problem in question is a special case of the classic {\sc $3$-Dimensional Matching},
which in turn is a special case of the {\sc $k$-Set Packing} problem, it is reasonable
to ask whether recent improvements in approximation algorithms for the latter
two problems [Cygan, FOCS'13; Sviridenko \& Ward, ICALP'13] allow for an improved
approximation for \textsc{Upper Degree-Constrained Partial Orientation}.
We follow this line of reasoning and present a polyno\-mial-time
local search algorithm with approximation ratio $5/4+\eps$.
Our algorithm uses a combination of two types of rules:
improving sets of bounded pathwidth from the recent
$4/3+\eps$-approximation algorithm for {\sc $3$-Set Packing} [Cygan, FOCS'13],
and a simple rule tailor-made for the setting of partial orientations.
In particular, we exploit the fact that one can check in polynomial time
whether it is possible to orient all the edges of a given graph [Gyárfás \& Frank, Combinatorics'76].
\end{abstract}

\section{Introduction}

\newcommand{\udpo}{\textsc{UDPO}\xspace}

During the last decades several graph orientation problems
were studied (see Section $8.7$ in~\cite{gutin} and Section $61.1$ in~\cite{schrijver}).
One of the most recently introduced is the \textsc{Upper Degree-Constrained Partial Orientation}, abbreviated as \udpo.
In the \udpo problem we are given an undirected graph $G=(V,E)$,
together with two degree constraint functions $d^-,d^+ : V \to \mathbb{N}$.
The goal is to orient as many edges as possible, in such a way
that for each vertex $v \in V$ the number of arcs entering $v$ is at most $d^-(v)$,
whereas the number of arcs leaving $v$ is at most $d^+(v)$.
This problem was introduced by Gabow~\cite{gabow}, motivated by a 
variant of the maximum bipartite matching problem arising when
planning a two-day event with several parallel sessions
and each participant willing to attend one chosen session
each day, but without a particular order on the two
selected sessions (for the exact definition, see~\cite{gabow}).

\renewcommand{\sp}{\textsc{set packing}\xspace}
\defproblem{\textsc{Upper Degree-Constrained Partial Orientation (UDPO)}}{Undirected graph $G$, degree constraints $d^+, d^- : V(G)\to \pint$}{%
A subset $\overline{F}\sub E(G)$ which admits an orientation $F$ satisfying 
$\deg^+_{F}(v)\le d^+(v)$ and $\deg^-_{F}(v)\le d^-(v)$ for each $v\in V(G)$.}{$|F|$}

Gabow proved the problem to be MAXSNP-hard (thus also APX-hard),
and showed an LP-based iterative rounding $4/3$-approximation algorithm.
As already observed by Gabow, \udpo is a special case of the $3$-{\sc Dimensional Matching} problem,
which in turn is a special case of $k$-\sp.
Both of these problems belong to the Karp's list of $21$ NP-complete problems,
and until last year the best known polynomial-time approximation algorithm
was due to Hurkens and Schrijver~\cite{hs} with approximation ratio $(k+\eps)/2$.
However this was recently improved independently by Sviridenko and Ward~\cite{sviridenko-ward} to $(k+2)/3$-approximation
and by Cygan~\cite{cygan} to $(k+1+\eps)/3$-approximation.
The latter result involves colour coding and pathwidth, tools originating from the area called {\em Fixed Parameter Tractability}, 
in local search routines.

\defproblem{$k$-\sp}{A family $\F$ of subsets of a finite universe $U$, such that $|F|\le k$ for every $F\in \F$}{%
A subfamily $\F_0\sub \F$ of pairwise-disjoint subsets}{$|\F_0|$}

\subsection{Our results}

Since $(k+1+\eps)/3$-approximation for $k$-\sp implies a $(4+\eps)/3$-approximation for \udpo,
one can ask whether recent developments for the former may be used to obtain an improved
algorithm for the latter.
In this paper we follow this line of reasoning and present a local search 
$(5+\eps)/4$-approximation algorithm, improving over the $4/3$-approximation ratio
of Gabow~\cite{gabow}.
In fact, our approximation ratio matches the $5/4$ lower bound on the integrality
gap of the natural LP relaxation obtained by Gabow~\cite{gabow}.

Our algorithm uses two types of rules trying to improve the current solution at hand.
Firstly, we invoke the bounded pathwidth local search by Cygan~\cite{cygan}
in a black-box manner, when treating the \udpo problem as an instance of $3$-\sp.
Secondly, we use a custom rule for \udpo, relying
on the fact that using a polynomial-time algorithm of Gyárfás \& Frank~\cite{gyarfas} one can check whether a
given set of undirected edges admits a feasible orientation
(satisfying the degree constraints).

In the analysis we focus on {\em simple} instances,
where all the degree bounds are either zero or one, which means that each vertex 
can have only zero or one incoming and outgoing arcs.
Interestingly, as shown in Section~\ref{sec:reduction},
for our local search routines simple instances are actually no easier than the arbitrary ones.

\subsection{Organization of the paper}

In the following subsection we discuss related work on the subject.
Next, in Section~\ref{sec:3d} we recall the reduction from \udpo
to $3$-\sp, followed by Section~\ref{sec:ksp}
with a description of basic notation for the local search algorithm from previous work on $k$-\sp.
Our algorithm is presented in Section~\ref{sec:alg}.
Its analysis on simple instances
(with all degree bounds at most one) is provided
in Sections~\ref{sec:ksp2} and~\ref{sec:analysis}, preceded, in Section~\ref{sec:reduction}, by a
reduction proving that the worst-case approximation ratio
is already attained by simple instances.

\subsection{Related work on $k$-\sp}

Between the algorithms of Hurkens and Schrijver and the recent improvements
for the $k$-\sp problem, quasipolynomial-time approximation algorithms were considered~\cite{h95,cgm13}.

There also is a line of research on the weighted variant of $k$-\sp,
where we want to select a maximum-weight family of pairwise-disjoint sets from $\F$.
Arkin and Hassin~\cite{arkin-hassin}
gave a $(k-1+\eps)$-approximation algorithm,
later Chandra and Halld{\'o}rsson~\cite{chandra-halldorson}
improved it to a $(2k+2+\eps)/3$-approximation.
Currently, the best-known approximation ratio is $(k+1+\eps)/2$ due to Berman~\cite{berman}.
All the mentioned results are based on local search.

For the standard (unweighted) $k$-\sp problem,
Chan and Lau~\cite{lau} also presented a strengthened
LP relaxation with integrality gap $(k+1)/2$.

On the other hand, Hazan et al.~\cite{hazan}
proved that $k$-\sp  is hard to approximate within a factor of $\Oh(k/\log k)$. 
Concerning small values of $k$, Berman and Karpinski~\cite{berman-karpinski}
obtained a $98/97-\eps$ hardness for $3$-{\sc Dimensional Matching}, which implies
the same lower bound for $3$-\sp.

\section{Preliminaries}
Let $G$ be an undirected (multi)graph. We sometimes treat $G$ as a directed graph,
where each \emph{edge} $e\in E(G)$ is represented by a pair of oppositely directed \emph{arcs} in $A(G)$.
For an arc $e\in A(G)$ we denote by $\overline{e}$ the corresponding edge in $E(G)$,
and by $e^R$, the reverse arc.
We also define $\overline{A}=\{\overline{e} : e\in A\}$ and $A^R=\{e^R:e\in A\}$ for an arbitrary subset $A\sub A(G)$.

A \emph{partial orientation} of $G$
can be defined as a subset $F\sub A(G)$ such that $F^R\cap F = \emptyset$.
It is called \emph{feasible} (for degree constraints $d=(d^+, d^-)$), if $\deg^+_F(v)\le d^+(v)$
and $\deg^-_F(v)\le d^-(v)$ for each $v\in V(G)$,
that is, if the number of arcs leaving $v$ and the number of arcs entering $v$ do not violate the upper bounds.
Now, \udpo can be reformulated as the problem of finding a maximum feasible partial orientation $F$,
rather than the corresponding set of undirected edges $\overline{F}$.

For an undirected (multi)graph $G$ and a set $U\sub V(G)$ we also define $N_G(U)$ as the set of vertices $v\notin U$ adjacent to some $u\in U$;
we also set $N_G[U] = N_G(U)\cup U$.

\subsection{Reduction to 3-\sp}\label{sec:3d}
The following reduction to 3-\sp was introduced by Gabow~\cite{gabow}.
Let $I=(G,d)$ be an instance of \udpo. We construct an equivalent instance of the 3-\sp problem,
i.e., a set family $\F$ over a universe $U$.

The universe $U$ is a disjoint union of three sets:  $V^+$, $V^-$ and $E$.
The set $V^+$ contains $d^+(v)$ copies $v^+_i$ of each $v\in V(G)$,
$V^-$ contains $d^-(v)$ copies $v^-_i$ of each $v\in V(G)$,
and $E$ is defined as~$E(G)$. 
The family $\F$ consists of sets $\{u^+_i, v^-_j,e\}$ and $\{v^+_j, u^-_i,e\}$ for each edge $e=\{u,v\}$ and
all possible indices $i,j$.

Given a feasible partial orientation $F$, the constraints clearly let us
choose for each arc $e=uv$ two copies $u^+_i$ and $v^-_j$,
so that the choices are distinct across all arcs leaving $u$ and entering $v$, respectively. 
Consequently, the sets $\{u^+_i, v^-_j,\overline{e}\}$ form a disjoint subfamily of $\F$.
Similarly, given any disjoint set-family $\F_0\sub \F$ it is easy to see
that orienting $e$ from $u$ to $v$ for any  $\{u^+_i, v^-_j,\overline{e}\}\in \F_0$ gives a feasible partial orientation.

\subsection{Local search for $k$-\sp}
\label{sec:ksp}

In this section we recall and reinterpret some of the results
behind the recent $\frac{k+1+\eps}{3}$-approximation algorithm by Cygan~\cite{cygan} for the $k$-\sp problem.

For an instance $(U,\F)$ of the $k$-\sp problem, we build an undirected conflict graph $G=G(\F)$ with $V(G)=\F$
and vertices $F,F'$ made adjacent if $F\cap F'\ne \emptyset$. 
Observe that solutions to this instance of $k$-\sp form independent sets in this graph.

The algorithm of~\cite{cygan} is based on the local-search principle.
It maintains a solution $\F_0\sub \F$ and tries to replace it with a larger, but similar solution.
It tries to use a disjoint family $X\sub \F\sm \F_0$ and replace $\F_0$ by $\F_0' =(\F\sm N_G(X))\cup X$,
where $G=G(\F)$ is the conflict graph. Note that $N_G(X)\cap \F_0$ consists exactly of those members of $\F_0$ 
which cannot be present together with $X$ in a single disjoint family.
It is reasonable to preform this operation if the resulting family $\F_0'$ is larger than $\F_0$,
or equivalently $|N_G(X)\cap \F_0| < |X|$.
This leads to a notion of \emph{improving sets}, defined for $\F_0\sub \F$ as disjoint families
$X\sub \F\sm \F_0$ such that $|N_G(X)\cap \F_0| < |X|$.

The classic approach to the $k$-\sp problem is to search for improving sets of sufficiently large constant size,
which leads to a $\frac{k+\eps}{2}$-approximation factor~\cite{hs}.
The novel idea of~\cite{cygan} was to consider larger improving sets satisfying structural
properties, which let us efficiently find these sets.
This is achieved using a structural parameter of a graph called \emph{pathwidth}.
In this paper we only use some results of~\cite{cygan} as a black-box, so we do not need to recall
the relatively complex definition of pathwidth.
Pathwidth of an undirected graph $G$, denoted as $\ptw(G)$,
does not exceed the number of vertices of $G$.
Pathwidth of an improving set $X$ is defined as $\ptw(G[N_G[X]])$ where $G=G(\F)$ is the conflict graph
and $G[N_G[X]]$ is the subgraph of $G$ induced by $N_G[X]$.
The following theorem uses techniques of fixed-parameter tractability to find improving sets of logarithmic size
and constant pathwidth in the conflict graph.
\begin{theorem}[\cite{cygan},Theorem 3.6]\label{thm:cygan}
There is an algorithm, that given a $k$-\textsc{set-packing} instance~$\F$,
and a disjoint family $\F_0\sub \F$, in $2^{\Oh(r\cdot k)}|\F|^{\Oh(pw)}$ time
determines whether there exists an improving set $X\sub \F \sm \F_0$ of size at most $r$ and pathwidth at
most $pw$, and if so, finds such an improving set.
\end{theorem}
\noindent
Finally, let us make an easy observation, stating that the algorithm is \emph{monotone} in a certain sense.
\begin{observation}\label{obs:monot}
If no improving set can be found using Theorem~\ref{thm:cygan} for $\F_0\sub \F$,
then one still cannot find an improving set if the 
instance $\F$ is restricted to any $\F'$ such that $\F_0\sub\F'\sub \F$.
\end{observation}

\section{Algorithm for \udpo}\label{sec:alg}
Our algorithm for \udpo combines the local-search rule by Cygan~\cite{cygan} for  
3-\sp, applied to an instance obtained through the reduction given in Section~\ref{sec:3d},
with a new custom rule.
This rule also tries to extend a feasible partial orientation $F$, but it works
with partial orientations as sets of undirected edges rather than directed arcs.
Given a partial orientation $F$ it tries to find a partial orientation $F'$ such that $|F'|>|F|$
and $\overline{F}'\Delta \overline{F}$, the symmetric difference between the underlying undirected
versions of $F$ and $F'$, is of constant size.
Polynomial time is sufficient to generate all possible choices of $\overline{F}'$, but it is not
enough to check all orientations $F'$. To overcome this issue, for $(V, \overline{F}')$ we apply a result
of Gyárfás and Frank, who used maximum-flow techniques to find in polynomial time
a (total) orientation satisfying degree constraints.
\begin{lemma}[\cite{gyarfas}]
Given an undirected graph $G$ and upper-degree constraints $d$, one can in polynomial
time decide whether there is a feasible partial orientation using all edges of $G$.
\end{lemma}
\begin{corollary}\label{cor:or}
There is an algorithm, that given a \udpo instance $(G,d)$ and a feasible partial orientation
$F$, in $\Oh(|E(G)|^{r}poly(|G|))$ time determines whether there exists a feasible partial
orientation $F'$ satisfying $|F'| > |F|$ as well as $|\overline{F}'\Delta \overline{F}|\le r$, and if so, finds such a feasible partial orientation.
\end{corollary}

We conclude this section with a succinct description of the algorithm.
Given an instance $I=(G,d)$ of \udpo, it builds an equivalent instance $\F$
of the 3-\sp problem using a reduction of Section~\ref{sec:3d}. It maintains a feasible partial orientation
$F$ together with a corresponding disjoint subfamily $\F_0\sub \F$,
while using the following two rules to improve $F$:
\begin{enumerate}
  \item\label{it:3d} apply Theorem~\ref{thm:cygan} to find an improving set for $\F_0$ of size at most $c_\eps \log |U|$  with pathwidth at most $c_\eps$, where $|U| = |E(G)| + \sum_{v \in V(G)} (d^+(v) + d^-(v))$ is the universe size of the 
  underlying instance of 3-\sp.
  \item\label{it:rev} apply Corollary~\ref{cor:or} to find a partial orientation $F'$ satisfying $|F'|>|F|$ and  $|\overline{F}'\Delta \overline{F}|\le c_\eps$. 
\end{enumerate}
The algorithm terminates if neither of the two rules is able to improve $F$.
Any such partial orientation $F$ is called a \emph{local optimum}.
The remaining part of this paper is devoted to analyzing how big the local optimum can be
compared to the global optimum. More precisely, we show that for every $\eps$
there is an appropriate choice of $c_\eps$ so that $|F|\ge (\frac45+\eps) |OPT|$ for
any local optimum $F$ and global optimum $OPT$.

\section{Reduction to simple instances}
\label{sec:reduction}

An instance $I=(G,d)$ of \udpo is called \emph{simple} if $d^+(v), d^{-}(v)\in \{0,1\}$ for every $v\in V(G)$
and \emph{proper} if $\deg_G(v)\ge \max(d^+(v),d^-(v))> 0$ for every $v\in V$.
Clearly, any instance can be easily reduced to an equivalent proper instance. 
In this section we show that it suffices to analyze our local-search algorithm for simple instances.

\begin{theorem}\label{thm:red}
Fix a constant $c_\eps>1$ for the algorithm of Section~\ref{sec:alg}.
Suppose that there exists an instance $I$ of \udpo with a locally-optimum partial orientation $F$
such that $|F|=\alpha |OPT_{I}|$. Then there exists a simple instance $I'$ of \udpo with a locally-optimum partial orientation $F'$ 
satisfying $|F'|=\alpha|OPT_{I'}|$.
\end{theorem}

Let $I=(G,d)$ be an arbitrary instance. For a pair of distinct non-adjacent vertices $u,v\in V(G)$
we define the operation of \emph{joining} $u$ and $v$ as follows: $u$ and $v$ are \emph{identified} in $G$ into a single vertex $w$
and their degree constraints for $w$ are obtain by summing the respective constraints for $u$ and $v$. Note that this operation preserves the set of edges.
Observe that in terms of the instance of 3-\sp obtained through the reduction of Section~\ref{sec:3d},
joining can be interpreted as introducing some sets to $\F$. Consequently, if a partial orientation is feasible in $I$, it is also feasible in the resulting instance $I'$, but the converse does not necessarily hold.

If $I'$ is obtained from $I$ by joining $u$ and $v$ into $w$, we say that $I$ can be obtained from $I'$
by \emph{splitting} $w$. Splitting is said to \emph{preserve} a partial orientation $A$,
if $A$ is feasible in $I'$ and remains feasible in $I$.

\begin{lemma}\label{lem:split}
Let $I=(G,d)$ be a proper instance with two feasible partial orientations $A,B$. 
If $\max(d^+(v),d^-(v))\ge 2$ for some $v\in V(G)$, then one can split $v$ so that both $A$ and $B$ 
are preserved and the resulting instance $I'$ is proper.
\end{lemma}
\begin{proof}
First, let us introduce an auxiliary vertex $v'$ connected to $v$ by $d^+(v)+d^-(v)$ parallel edges.
We extend $d$ to $v'$ setting the constraints large enough to accommodate all edges incident to $v'$.
Note that this operation has no effect on whether one can split $v$.

Now, let us modify $A$ to obtain $A'$ by orienting $d^+(v)-\deg^+_A(v)$ edges from $v$ to $v'$ and $d^-(v)-\deg^-_A(v)$ 
edges from $v'$ to $v$. Note that $A'$ is feasible in the extended graph and the degree constraints for $v$
are tight. Analogously, we extend $B$ to $B'$. 
A larger partial orientation may only be harder to preserve, so it suffices
to prove that one can split $v$ preserving $A'$ and $B'$.
Equivalently, the construction in this paragraph lets us assume that $\deg^+_A(v)=\deg^+_B(v)=d^+(v)$ and $\deg^-_A(v)=\deg^-_B(v)=d^-(v)$.

Both for $A$ and $B$ we classify edges of $G$ incident to $v$ into three types:
oriented towards $v$ ($-$), oriented towards the other endpoint ($+$) and not included in the orientation ($0$).
In total, we get a partition of the set $\delta(v)$, consisting of edges incident to $v$, into nine sets $E_{ab}$ with $a,b\in \{+,-,0\}$; here $a$ corresponds to the orientation in $A$
and $b$ to the orientation in~$B$.

In some situations, one can clearly take a few edges incident to $v$, and split $v$ into two vertices,
one \emph{new} vertex $v'$ incident to the selected edges, and the other, still denoted as $v$, incident to the remaining edges.
We refer to this operation as \emph{splitting out} some edges.
Note that in order to preserve both $A$ and $B$, we need to split out edges so that for $v'$
the number incoming edges
is the same in both orientations, similarly for the outgoing arcs.
We shall make sure that this number is always 0 or 1, i.e., $(d^+(v'),d^-(v'))\in \{(0,1),(1,0),(1,1)\}$. The constraints at $v$ are decreased accordingly.
\begin{enumerate}
  \item\label{it:1} If $E_{++}\ne \emptyset$, one can split out a single edge $e\in E_{++}$ setting constraints $(1,0)$;
  symmetrically if $E_{--}\ne \emptyset$ one sets $(0,1)$.
  \item\label{it:2} If $E_{+-},E_{-+}\ne \emptyset$, one can split out two edges -- one of each type, setting constraints $(1,1)$.
  \item\label{it:3} If $E_{0+},E_{+0}\ne \emptyset$, one can split out two edges -- one of each type, setting constraints $(1,0)$;
  symmetrically if $E_{0-},E_{-0}\ne \emptyset$ one sets $(0,1)$.
  \item\label{it:4} If $E_{+-}, E_{0+}, E_{-0}\ne \emptyset$ one can split out three edges -- one of each type, setting constraints $(1,1)$;
  symmetrically if $E_{-+}, E_{+0},E_{0-}\ne \emptyset$ one also sets $(1,1)$.
\end{enumerate}

We shall prove that one of these rules is always applicable. Note that the resulting instance is guaranteed to be proper
as we have $\max(d^+(v), d^-(v)) \ge 2$, so it is impossible to leave $v$ with both constraints equal to $0$,
which is forbidden in proper instances.

We proceed by contradiction, showing that if no rule is applicable, then $d^+(v)=d^-(v)=0$, which is impossible because
 $I$ is proper.
Let $n_{ab}=|E_{ab}|$. Recall that we have made an assumption that $\deg^+_A(v)=\deg^+_B(v)=d^+(v)$ and $\deg^-_A(v)=\deg^-_B(v)=d^-(v)$, which implies the following equalities:
\begin{align*}
n_{0+}+n_{++}+n_{-+}&=d^+(v)=n_{+0}+n_{++}+n_{+-},\\
n_{0-}+n_{+-}+n_{--}&=d^-(v)=n_{-0}+n_{-+}+n_{--}.
\end{align*}
If $n_{++}>0$ or $n_{--}>0$ we could apply rule \ref{it:1}.
Therefore
\begin{align*}
n_{0+}+n_{-+}&=d^+(v)=n_{+0}+n_{+-},\\
n_{0-}+n_{+-}&=d^-(v)=n_{-0}+n_{-+}.
\end{align*}
If $n_{+-}>0$ and $n_{-+}>0$ we could apply rule \ref{it:2};  without loss of generality we assume $n_{+-}=0$ and thus
\begin{align*}
n_{0+}+n_{-+}=d^+(v)=n_{+0},\\
n_{0-}=d^-(v)=n_{-0}+n_{-+}.
\end{align*}
Consequently, we have $n_{+0}\ge n_{0+}$ and $n_{0-}\ge n_{-0}$.
Therefore, if $n_{0+}>0$ or $n_{-0}>0$, we could apply rule \ref{it:3},
which means that both these values are equal to 0 and
$$n_{0-}=n_{+0}=n_{-+}=d^+(v)=d^-(v).$$
However, if the common value of these variables was not equal to 0, we could apply rule \ref{it:4}.
This way we get the announced contradiction.
\end{proof}

\begin{corollary}\label{cor:split}
If $I$ is a proper instance with feasible partial orientations $A$ and $B$, then
with a finite sequence of vertex splitting preserving both $A$ and $B$, one can obtain a simple proper instance $I'$.
\end{corollary}
\begin{proof}
It suffices to exhaustively apply Lemma~\ref{lem:split}. Observe that this process
must terminate, as vertex splitting increases the number of vertices and changes neither $D^+=\sum_{v\in V(G)}d^+(v)$ nor  $D^-=\sum_{v\in V(G)}d^-(v)$, while $|V(G)|\le D^++D^-$ for any proper instance,
\end{proof}

For a proof of Theorem~\ref{thm:red}, it suffices to apply Corollary~\ref{cor:split} for $A=F$ and $B=OPT_I$.
Vertex splitting may only reduce the family of feasible partial orientations, so $OPT_I$ is still a global optimum.
Also, this operation preserves $F$ as a local optimum with respect to rule~\ref{it:rev}.
For rule~\ref{it:3d} the analogous property follows 
from the fact that vertex splitting can be seen as removing sets in the underlying instance of 3-\sp
(without changing the size of the universe), and by Observation~\ref{obs:monot}, the corresponding rule for 3-\sp is monotone,
i.e., removing sets from the universe does not make finding an improving set easier.

Therefore, Corollary~\ref{cor:split} gives a simple instance  $I'$
for which $F$ and $OPT_I$ are still a local and a global optimum, respectively. 

\section{Tools from $k$-\sp}
\label{sec:ksp2}
In this section we recall and reinterpret several pieces of the analysis
of the local search algorithms for $k$-\sp, see~\cite{hs,cygan}.

This analysis focuses on the subgraph of the conflict graph $G(\F)$ induced by two solutions:
a local and a global optimum. Sets belonging to both families can be ignored, which leads to
a bipartite graph with degrees bounded by $k$. The following results
are stated in the language of abstract bipartite graphs, so that we can also use them in a slightly
different context.

\begin{definition}
Let $H=(A,B,E(H))$ be a bipartite graph. A set $X\sub B$ is called \emph{improving},
if $|N_H(X)|<|X|$.
\end{definition}

The following lemma is a part of the analysis
of the classic  $(k+\eps)/2$-approximation local search, which goes back to Hurkens and Schrijver~\cite{hs}.
Our proof is based on the proof of Lemma 3.11 in~\cite{cygan}. Although that result uses larger class of improving sets
to obtain a better bound on $\frac{|B|}{|A|}$, the overall line of reasoning remains the same.
\begin{lemma}\label{lem:ksetp}
Fix a positive integer $k\ge 3$. For any $\eps>0$ there exists a constant $c_\eps$ satisfying the following property.
Let $H=(A,B,E(H))$ be a bipartite graph with degrees not exceeding $k$.
If there is no improving set $X\sub B$ with $|X|\le c_\eps$, then $|B|\le \frac{k+\eps}{2}|A|$.
\end{lemma}
\begin{proof}
We are going to construct a sequence of at most $\frac{1}{\eps}$ induced subgraphs $H_i=H[A_i,B_i]$,
with $A_i\sub A$ and $B_i \sub B$.
These subgraphs shall satisfy the following two properties:
\begin{enumerate}[(a)]
  \item\label{it:rimp} in $H_i$ there is no subset $X\sub B_i$ such that $|X|\le 2(k+1)^{\frac{1}{\eps}-i}$ and $|N_{H_i}(X)|< |X|$,
  \item\label{it:eq} $|A\sm A_i| = |B\sm B_i|\ge \eps i|A|$.
\end{enumerate}
We start with $H_0=H$, which trivially satisfies~(\ref{it:eq}). It suffices to take $c_\eps=2(k+1)^{\frac{1}{\eps}}$
to make sure that~(\ref{it:rimp}) also holds.

Consider the graph $H_i$. Let us classify vertices of $B_i$ based on their degree in $H_i$: we define $B_i^d$
as the set of vertices of degree $d$, and $B_i^{d+}$ as the set of vertices of degree at least $d$.
Note that~(\ref{it:eq}) implies $i \le \frac{1}{\eps}$, and thus $2(k+1)^{\frac{1}{\eps}-i} \ge 2$.
Consequently, by~(\ref{it:rimp}), $B_i^0=\emptyset$  and the vertices of $B_i^1$ have distinct neighbors (otherwise we would have an improving
set of size one or two, respectively).

We consider two cases, depending on whether $|B_i^1|\le \eps |A|$.
First, we suppose this inequality does not hold.
Then we construct $H_{i+1}$ setting $B_{i+1}=B_i^{2+}$ and $A_{i+1}=A_i\sm N_{H_i}[B_i^1]$. As we have noted,
vertices in $B_i^1$ do not share neighbours, so $|A_i\sm A_{i+1}|=|B_i^1|=|B_i\sm B_{i+1}|$,
and consequently $|B\sm B_{i+1}|=|A\sm A_{i+1}|$. Also, we  clearly have $|B\sm B_{i+1}|\ge \eps i|A|+|B_i^1|\ge \eps(i+1)|A|$.

Therefore, it suffices to show that $H_{i+1}$ satisfies property (\ref{it:rimp}).
Take $X\sub B_{i+1}$ such that $|N_{H_{i+1}}(X)|<|X|$. We construct $X'\sub B_i$
with $|N_{H_i}(X')|< |X'|$ such that $|X'|\le (k+1)|X|$. 
Clearly, if $X$ then contradicts (\ref{it:rimp}) for $H_{i+1}$, so does $X'$ for $H_i$.
Recall that $H_i[B_i\sm B_{i+1}, A_i\sm A_{i+1}]$ is a perfect matching.
We denote the unique neighbor of a vertex $v$ in this graph by $m(v)$.
We simply define $X' = X \cup \{m(a) : a\in (A_{i}\sm A_{i+1})\cap N_{H_i}(X)\}$ (see also Figure~\ref{fig1}).
Then $N_{H_i}(X')=N_{H_i}(X) = N_{H_{i+1}}(X)\cup \{m(b) : b\in X'\sm X\}$.
Consequently, $|N_{H_i}(X')|=|N_{H_{i+1}}(X)|+|X'\sm X|<|X|+|X'\sm X|=|X'|$.
Moreover, by the degree restriction in $H$, we have $|N_{H_i}(X)|\le k|X|$,
and thus $|X'|\le |X|+|N_{H_i}(X)|\le (k+1)|X|$, as claimed.

\begin{figure}[ht]
\begin{center}
\begin{tikzpicture}[scale=1]
  \tikzstyle{vertex}=[circle,fill=black,minimum size=0.20cm,inner sep=0pt]
  \tikzstyle{vertex2}=[circle,draw=black,fill=gray!50,minimum size=0.20cm,inner sep=0pt]
  \tikzstyle{terminal}=[rectangle,draw=black,fill=white,minimum size=0.2cm,inner sep=0pt]

	\foreach \x in {2,3,4,5}
	{
		\node[vertex] (a\x) at (\x,2){};
	}
	\node[vertex2] (a0) at (0,2){};
	\node[vertex2] (a1) at (1,2){};

	\foreach \x in {2,3,4}
	{
		\node[vertex] (b\x) at (\x+0.5,0){};
	}
	\node[vertex2] (b0) at (0.5,0){};
	\node[vertex2] (b1) at (1.5,0){};

	\draw (a1) -- (b1);
	\draw (a0) -- (b0);
	\draw (a2) -- (b2);
	\draw (a2) -- (b4);
	\draw (a2) -- (b0);
	\draw (a2) -- (b1);
	\draw (a3) -- (b1);
	\draw (a3) -- (b4);
	\draw (a3) -- (b3);

\draw (a3) -- (b2);
\draw (a4) -- (b4);
\draw (a4) -- (b3);
\draw (a5) -- (b4);
\draw (a5) -- (b3);

\draw[decorate, decoration=brace] (1.8,2.2) -- (5.2,2.2);
\draw (3.5,2.6) node {$X$};

\draw[decorate, decoration=brace] (-0.2,2.2) -- (1.2,2.2);
\draw (0.5,2.6) node {$X' \setminus X$};

\draw[decorate, decoration=brace] (4.7,-0.2) -- (2.3,-0.2);
\draw (3.5,-0.6) node {$N_{H_{i+1}}(X)$};

\end{tikzpicture}
\end{center}
\caption{Lifting an improving set $X$ in $H_{i+1}$ to an improving set $X'$ in $H_{i}$. Gray vertices belong to $H_i$ but not to $H_{i+1}$.}
\label{fig1}
\end{figure}
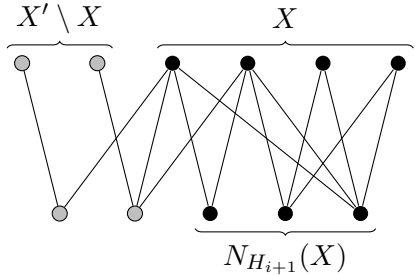

Therefore it suffices to consider the case when $|B_i^1|\le \eps |A|$.
We count edges of $H_i$;  clearly, $|E(H_i)|\le k|A_i|$ since the degrees do not exceed $k$.
On the other hand, $|E(H_i)|\ge |B_i^1|+2|B_i^{2+}|$, and consequently $|B_i^1|+2|B_i^{2+}|\le k|A_i|$.
Summing up, we get
$$2|B| = 2|B\sm B_i| + 2|B_i| = 2|A\sm A_i| + 2|B_i^1|+2|B_i^{2+}|
\le 2|A\sm A_i| +|B_i^1|+k|A_i|\le (k+\eps)|A|,$$
that is, $|B|\le \frac{k+\eps}{2}|A|$, which concludes the proof.
\end{proof}

The following lemma is, on the other hand, a slight generalization of Lemma 3.11 in~\cite{cygan},
restricted to $k=3$. Under the original assumptions it shows that the ratio $\frac{|B|}{|A|}$
is close to the worst-case $\frac{4}{3}$ only if (almost) all vertices in $A$ are of degree 3, and thus allows for a better
bound if some fraction of vertices have degree at most 2. 
\begin{lemma}\label{lem:ptw}
For any $\eps>0$ there exists a constant $c_\eps$ satisfying the following property.
Let $H=(A,B,E(H))$ be a bipartite graph with degrees not exceeding $3$.
If there is no improving set $X\sub B$ such that $|X|\le c_\eps \log |V(H)|$ and $\ptw(H[N_G[X]])\le c_\eps$,
then $$|B|\le (1+\eps)|A|+\tfrac13|\{a\in A : \deg_H(a)\ge 3\}|.$$
\end{lemma}
\begin{proof}
We follow the notation and the main line of reasoning of the proof of Lemma~\ref{lem:ksetp}, which
for $k=3$ has stronger requirements for $X$.
We only alter the last step of the proof,
i.e., the analysis when $|B_i^1|\le \eps|A|$.
This requires the following reformulation of Claim 3.12 from~\cite{cygan},
which is where we use the whole strength of the assumptions of Lemma~\ref{lem:ptw}.
\begin{claim}[\cite{cygan}]
For large enough $c_\eps$ we have $|B_i^2| \le (1+\eps)|A_i|$. 
\end{claim}

As before, we count edges $E(H_i)$.
We clearly have $|E(H_i)|= |B_i^1|+2|B_i^{2}|+3|B_i^{3}|$.
On the other hand, $|E(H_i)|\le 2|A_i|+|A_i^3|$ where $A_i^3=\{a\in A_i : \deg_{H_i}(a)=3\}$.
Summing up, we obtain
\begin{multline*}
3|B|=3|B\sm B_i|+3|B_i^1|+3|B_i^2|+3|B_i^{3}|= 3|A\sm A_i|+2|B_i^1|+|B_i^2|+|E(H)|\le\\
3|A\sm A_i|+2\eps|A|+(1+\eps)|A_i|+2|A_i|+|A_i^3|\le 3(1+\eps)|A|+|\{a\in A: deg_H(a)=3\}|,
\end{multline*}
that is, $|B|\le (1+\eps)|A|+|\{a\in A: deg_H(a)=3\}|$, which completes the proof.
\end{proof}

\section{Analysis}
\label{sec:analysis}
\newcommand{\opt}{\overline{OPT}}
\renewcommand{\a}{\overline{A}}
\renewcommand{\b}{\overline{B}}
\newcommand{\f}{\overline{F}}

We start the analysis of the algorithm of Section~\ref{sec:alg}
with a result which lets us construct the counterpart of the bipartite
conflict graph with respect to two feasible solutions.
Later we apply Theorem~\ref{thm:red}, which allows restricting to simple instances.

\begin{lemma}\label{lem:conf}
Let $I$ be a simple instance of $\udpo$ and let $A,B$ be a pair of feasible partial orientations.
There exists a bipartite graph $H=(\b\sm \a,\a\sm \b, E(H))$ such that:
\begin{enumerate}[(a)]
  \item\label{it:deg} degrees in $H$ do not exceed 4,
\item\label{it:imp} for any $X\sub \b\sm \a$ there is a feasible partial orientation $F$ with $\overline{F}=(\overline{A}\sm N_H(X))\cup X$.
\end{enumerate}
\end{lemma}
\begin{proof}
Let $A'=\a\sm \b$, $B'=\b\sm \a$ and $G'=(V(G), \a\cap \b)$.
For a connected component $C$ of $G'$ we define $\delta_G[C]$ as the set of edges $e\in E(G)$ incident to at least one vertex of $C$.
We construct the graph $H$ as follows.
We make $a\in A'$ adjacent in $H$ to $b\in B'$ if and only if both $a$ and $b$ belong to $\delta_G[C]$ for some connected component $C$.

Let us prove that $H$ satisfies the desired properties, starting with~(\ref{it:deg}).
Consider any connected component $C$ of $G'$.
As $I$ is a simple instance, all the vertices in $G'$ are of degree at most two,
which means that $C$ is either a path or a cycle.
Consequently, in either case, again by the assumption that $I$ is simple,
we have $|\delta_G[C]\cap A'|\le 2$ and $|\delta_G[C]\cap B'|\le 2$,
because, both in $A$ and in $B$, at most $2|C|$ arc endpoints can be incident to $C$.
Any edge is incident to at most two components of~$G$, for each of them we may have created
at most two neighbors in $H$, and thus the degrees in $H$ are at most 4.

To prove~(\ref{it:imp}) we take $X\sub B'$ and consider a set $\overline{F}=(\overline{A}\sm N_H(X))\cup X$.
Note that for any component $C$ of $G'$ we have $\overline{F}\cap \delta_G[C]\sub \overline{A}\cap \delta_G[C]$
(if $X\cap \delta_G[C]=\emptyset$) or $\overline{F}\cap \delta_G[C]\sub \overline{B}\cap \delta_G[C]$ (otherwise).
We can orient edges of $\delta_G[C]\cap \overline{F}$ consistently with $A$ in the former case 
and consistently with $B$ in the latter. 
Note that if there is an edge $e\in \overline{F}$ between two connected components of $G'$,
then $e\notin \overline{A}\cap \overline{B}$, so both components
are oriented consistently with $A$ (if $e\in \overline{A}$) or $B$ (if $e\in \overline{B}$),
hence the proposed orientation is well-defined.
It remains to argue that if we orient the edges in this manner,
then all the capacity constraints are satisfied.
Consider any vertex $v$ of $G'$. As it belongs to exactly one connected component of $G$,
its incident edges from $\overline{F}$ are either oriented as in $A$ or as in $B$,
in either case the degree constraints are obeyed.
\end{proof}

Next, we apply the conflict graph and the technique similar to the standard analysis of the $(2+\eps)$-local search approximation
of $4$-set packing. This lets us derive a bound with respect to rule~\ref{it:rev}.

\begin{lemma}\label{lem:2eps}
Fix $\eps>0$. There exists a constant $c_{\eps}$ such that
for any simple instance~$I$ of \udpo the following condition holds.
Let $F$ be a feasible partial orientation which cannot be improved using rule~\ref{it:rev}
and let $OPT$ be an  optimum partial orientation. Then $|\opt\sm \f|\le (2+\eps)|\f\sm \opt|$.
\end{lemma}
\begin{proof}
We set $c_\eps$ as in Lemma~\ref{lem:ksetp} for $k=4$, and proceed with a proof by contradiction.
Suppose that $|\opt\sm \f|> (2+\eps)|\f\sm \opt|$.
We apply Lemma~\ref{lem:conf} to $A=F$ and $B=OPT$ to obtain a bipartite graph $H$, which
we plug to Lemma~\ref{lem:ksetp}. This implies that there is a set $X\sub V$ of size at most $c_\eps$
with $|N_H(X)|<|X|$. By Lemma~\ref{lem:conf}(\ref{it:imp}), replacing $X$ with $X'$
gives a feasible orientation, and rule~\ref{it:rev} would actually be able to perform this improvement. This contradicts the assumption that $F$ is a local optimum.
\end{proof}

Finally, we combine the consequences of rule~\ref{it:rev} (Lemma~\ref{lem:2eps})
with the strengthened analysis of rule~\ref{it:3d} (Lemma~\ref{lem:ptw})
to derive the main result of this paper.

\begin{theorem}\label{thm:local}
Fix $\eps>0$. There exists a constant $c_\eps$ such that 
for any instance of \udpo and any feasible partial orientation $F$ which cannot be improved using rules~\ref{it:3d} and~\ref{it:rev},
we have $|OPT|\le (\tfrac54+\eps)|F|$, where $OPT$ is a maximum feasible partial orientation.
\end{theorem}
\begin{proof}
By Theorem~\ref{thm:red}, it suffices to prove the claim for simple instances only.
Let $C=OPT\cap F$. Note that $F\sm C$ and $OPT\sm C$ induce a bipartite subgraph $H=(F\sm C, OPT\sm C, E(H))$ of the conflict graph in the underlying instance of 3-\textsc{set packing}. 
Clearly, the degrees in $H$ are bounded by 3.
Moreover, by construction of the reduction, if $\deg_H(e)=3$ for some $e\in F$,
then $e^R\in OPT$, i.e., $|\{e\in F : \deg_H(e)=3\}|\le |\overline{OPT}\cap \overline{F}|$

We set $c_\eps$ large enough for Lemmas~\ref{lem:ptw} and~\ref{lem:2eps} to be applicable.
The former lets us conclude that 
\begin{multline*}
|OPT|=|C|+|OPT\sm C| \le |C|+(1+\eps)|F\sm C|+\tfrac13|\{e\in F\sm C: \deg_H(e) = 3\}\le \\ (1+\eps)|F|+\tfrac13 |\overline{OPT}\cap \overline{F}|.
\end{multline*}
If $|\overline{OPT}\cap \overline{F}|\le \frac{3}{4} |F|$, this already concludes the proof.
Otherwise $|\overline{F}\sm \overline{OPT}|\le \frac14|F|$ and we apply Lemma~\ref{lem:2eps} to get
$$|\overline{OPT}\sm \overline{F}| \le (2+\eps)|\overline{F} \sm \overline{OPT}|,$$
and consequently we obtain
\begin{multline*}
|OPT| = |\overline{OPT}\sm \overline{F}|+|\overline{OPT}\cap \overline{F}| \le (2+\eps)|\overline{F}\sm \overline{OPT}|+|\overline{OPT}\cap \overline{F}| = (1+\eps)|\overline{F}\sm \overline{OPT}| +|F|\le \\ \tfrac{5+\eps}{4}|F|\le (\tfrac{5}{4}+\eps)|F|,
\end{multline*}
which concludes the proof.
\end{proof}

\subsection*{Acknowledgement}
We thank Harold N. Gabow for pointing us to the \udpo problem.

\bibliographystyle{plain}
\bibliography{orientation}
\end{document}